\documentclass[sigconf]{aamas}

\setcopyright{ifaamas}
\acmDOI{}
\acmISBN{}
\acmConference[AAMAS'27]{Proc.\@ of the 26th International Conference on
Autonomous Agents and Multiagent Systems (AAMAS 2027)}{2027}{}
\acmYear{2027}
\copyrightyear{2027}
\acmPrice{}

\usepackage{booktabs}
\usepackage{multirow}
\usepackage{amsmath}
\usepackage{graphicx}
\usepackage{xcolor}
\usepackage{subcaption}
\usepackage{url}
\usepackage{float}

\begin{document}

\title{Anchor-and-Resume Concession Under Dynamic Pricing for LLM-Augmented Freight Negotiation}

\author{Hoang Nguyen}
\affiliation{
  \institution{Georgia Institute of Technology}
  \city{Atlanta}
  \state{Georgia}
  \country{USA}
}
\email{hnguyen433@gatech.edu}

\author{Lu Wang}
\affiliation{
  \institution{Transportation Insight / NTG}
  \city{Atlanta}
  \state{Georgia}
  \country{USA}
}
\email{luwang@t-insight.com}

\author{Marta Gaia Bras}
\affiliation{
  \institution{Transportation Insight / NTG}
  \city{Atlanta}
  \state{Georgia}
  \country{USA}
}
\email{mbras@t-insight.com}

\begin{abstract}
Freight brokerages negotiate thousands of carrier rates daily under dynamic pricing conditions where models frequently revise targets mid-conversation. Classical time-dependent concession frameworks use a fixed shape parameter $\beta$ that cannot adapt to these updates. Deriving $\beta$ from the live spread enables adaptation but introduces a new problem: a pricing shift can cause the formula to retract a previous offer, violating monotonicity and signaling bad faith. LLM-powered brokers offer flexible negotiation but require expensive reasoning models, produce non-deterministic pricing decisions, and remain vulnerable to prompt injection from carrier messages.

We propose a two-index anchor-and-resume framework that addresses both limitations. A spread-derived $\beta$ automatically maps each load's margin structure to the correct concession posture, while the anchor-and-resume mechanism guarantees monotonically non-decreasing offers under arbitrary pricing shifts. All pricing decisions remain in a deterministic formula; the LLM, when used, serves only as a natural language translation layer. Empirical evaluation across 115{,}125 negotiations shows that the adaptive $\beta$ tailors behavior by regime: in narrow spreads, the framework concedes quickly to prioritize deal closure, carrier retention, and load coverage; in medium and wide spreads, it matches or exceeds the best fixed-$\beta$ baselines in broker savings. Against an unconstrained 20-billion-parameter LLM broker, the framework achieves similar agreement rates and savings. Against LLM-powered carriers as more realistic stochastic counterparties, the framework maintains comparable savings and higher agreement rates than against rule-based opponents, confirming that performance generalizes beyond scripted scenarios. By decoupling the LLM from pricing logic, the framework scales horizontally to thousands of concurrent negotiations with negligible inference cost and transparent decision-making.
\end{abstract}


\maketitle

\section{Introduction}
\label{sec:intro}

The US freight brokerage market is valued at \$19.7 billion in 2025~\cite{mordor2025}. Brokers act as intermediaries between shippers and carriers, negotiating rates across thousands of loads daily. For each load, the broker has a price band $[r_{\min}, r_{\max}]$ set by an internal pricing model and must negotiate a settlement rate with the carrier. The broker's objective is to minimize the agreed rate (maximizing margin), while the carrier seeks to maximize it. In production, these pricing models ingest live market signals and can update their estimates at any time. For example, the 2026 Strait of Hormuz disruption saw diesel futures surge 14\% in a single day as commercial shipping through the strait halted~\cite{hormuz2026}, forcing brokerages to revise carrier rate targets across hundreds of active negotiations simultaneously. A negotiation agent that cannot adapt its concession behavior in response to such updates will either overpay on deals where the target has tightened or lose deals where it has loosened.

Faratin et al.~\cite{faratin1998} formalized the concession curve through a shape parameter $\beta$ that determines how aggressively the agent concedes over time (Section~\ref{sec:adaptive_beta}). However, a single fixed $\beta$ cannot be universally optimal. Loads with narrow price spreads (1--2\%) leave little room, so the broker should concede quickly to secure coverage. Loads with wide spreads (15\%+) offer substantial room, so the broker should hold firm. The problem is compounded in production settings where the pricing model can update its estimate \emph{during} an active negotiation, changing the effective spread. A fixed $\beta$ cannot respond to these updates. An alternative is to delegate pricing decisions entirely to an LLM, but this introduces per-round inference cost that limits horizontal scaling, non-deterministic outputs that complicate auditability, and operational dependence on a single model provider.

We make four contributions:

\begin{enumerate}
    \item \textbf{Spread-derived $\beta$ for scalable deployment.} We replace Faratin's fixed $\beta$ with $\beta = c / (s \times 100)$, where $s$ is the fractional spread from the pricing model. This automatically maps each load's margin structure to the correct concession posture (Conceder for narrow spreads, Boulware for wide spreads), eliminating per-negotiation tuning and enabling platform-wide deployment across hundreds of thousands of annual deals (Section~\ref{sec:adaptive_beta}).

    \item \textbf{Two-index anchor-and-resume framework.} When a pricing shift occurs mid-negotiation, the framework decouples the negotiation round from the Faratin curve position, allowing the agent's concession behavior to adapt while guaranteeing monotonically non-decreasing offers. We prove this construction prevents retractions under arbitrary shift sequences, preserving good-faith negotiation and dampening pricing pipeline volatility from the carrier's perspective (Section~\ref{sec:two_index}).

    \item \textbf{LLM-agnostic architecture.} The deterministic strategy engine computes all pricing decisions; the LLM serves only as a natural language translation layer. This decoupling enables operators to switch LLM providers if one fails, use a cheaper model without affecting negotiation quality, and maintain full auditability of every pricing decision (Section~\ref{sec:llm_results}).

    \item \textbf{Empirical validation on business metrics.} We evaluate on three metrics that map to business outcomes: margin capture (broker savings), carrier retention (agreement rate), and system throughput (negotiation rounds). In three experiments totaling 115{,}125 negotiations under dynamic pricing, the two-index strategy achieves zero retractions across all conditions: a rule-based evaluation (105{,}000 negotiations across 12 spread values), a comparison against a 20-billion-parameter unconstrained LLM broker (3{,}375 negotiations), and a robustness test against LLM-powered carrier agents (6{,}750 negotiations). The framework matches LLM performance while operating as a lightweight deterministic formula with negligible inference cost, simple prompt instructions, and full auditability (Sections~\ref{sec:results}--\ref{sec:llm_results}).
\end{enumerate}

\section{Related Work}
\label{sec:related}

\subsection{Automated Negotiation and the Beta Problem}

Faratin et al.~\cite{faratin1998} introduced a taxonomy of negotiation tactics for autonomous agents, with time-dependent tactics using a concession function $\alpha(t) = t^{1/\beta}$ where $\beta$ controls the curve shape. The original work left $\beta$ as a design-time constant chosen by the system builder, explicitly identifying its selection as an open problem. Over the following 25 years, the field developed several approaches to address this gap, with modern methods largely focusing on adapting $\beta$ based on \emph{opponent behavior}.

Cao et al.~\cite{cao2015} proposed a multi-strategy selection algorithm that classifies the opponent's concession rate as accelerating, decelerating, or steady, and adjusts the agent's concession curve accordingly. Hindriks and Tykhonov~\cite{hindriks2008} applied Bayesian learning to estimate the opponent's preference profile (issue weights and evaluation functions) from observed bids, then used the learned model to select Pareto-efficient offers at the agent's own target utility level. The ANAC competition~\cite{baarslag2013} produced agents such as HardHeaded and Gahboninho that model opponent preferences from observed bids and dynamically adjust their concession schedules. More recently, reinforcement learning approaches~\cite{bagga2020} treat the concession parameter as a learnable variable optimized over thousands of simulated negotiations. Fatima et al.~\cite{fatima2014} provide a comprehensive treatment of the game-theoretic foundations underlying these approaches.

These methods adapt negotiation behavior based on \emph{opponent} actions and generally assume static reservation values. Our framework instead derives $\beta$ from the \emph{pricing model's} assessment of market conditions, which can change mid-negotiation when an external system updates parameters. When $\beta$ is derived from the live spread, such shifts can cause the concession formula to produce an offer lower than the previous one, violating the monotonic concession protocol. Endriss~\cite{endriss2006} formalized monotonic concession as a protocol requirement for bilateral bargaining. Winoto~\cite{winoto2005} introduced a non-monotonic-offers protocol and showed it generates higher surplus under full rationality, but the gains diminish as agent rationality decreases. In freight brokerage, where neither party observes the other's reservation value, retractions are interpreted as bad faith. Gear et al.~\cite{gear2020} addressed this problem directly with PredictRV, which uses statistical prediction models (Bayesian regression, LSTM) to smooth offers when the agent's reservation value changes during negotiation. Their approach reduces erratic offer behavior but provides no formal monotonicity guarantee, and it handles changes to the reservation value only while the concession shape parameter remains fixed. Our two-index framework guarantees monotonicity under arbitrary shift sequences that change both the concession range and the shape parameter simultaneously.

\subsection{AI-Driven Negotiation Agents}

A separate line of work applies neural and language models directly to negotiation. Lewis et al.~\cite{lewis2017} introduced end-to-end neural negotiation on the Deal-or-No-Deal task. Fu et al.~\cite{fu2023} showed that LLM agents can improve through self-play and in-context learning from AI feedback. Bianchi et al.~\cite{bianchi2024} introduced NegotiationArena, finding that behavioral tactics (e.g., pretending to be desperate) can improve payoffs by 20\%. Abdelnabi et al.~\cite{abdelnabi2024} introduced a multi-agent negotiation testbed and found that GPT-4's deal success rates degrade significantly when greedy or adversarial agents are present. In the supply chain domain specifically, Kirshner et al.~\cite{kirshner2026} found that LLM agents exhibit human-like bargaining behavior but are more inclined toward agreement, improving efficiency at the cost of equity.

These works treat the LLM as the decision-maker: the model chooses both the language and the negotiation posture. He et al.~\cite{he2018} showed that decoupling high-level strategy from language generation avoids the degeneracy of end-to-end reinforcement learning while achieving higher agreement rates. Our framework builds on this principle, replacing their learned dialogue acts with a deterministic strategy engine grounded in Faratin's concession theory. The strategy engine computes all numeric decisions; the LLM, when used, handles language generation only.

\section{Framework}
\label{sec:framework}

\subsection{Bilateral Negotiation Model}

We model freight rate negotiation as a single-issue bilateral alternating-offers game between a broker agent $B$ and a carrier agent $C$, where the sole negotiated variable is the linehaul rate. Each negotiation concerns a single load $\ell$ with parameters $r_{\min}$ (minimum rate), $r_{\max}$ (maximum rate), and an optional $r_{\text{target}}$ from the pricing pipeline. The zone of possible agreement (ZOPA) is $[r_{\min}, r_{\max}]$. The negotiation proceeds for up to $T$ rounds, ending when one side accepts, walks away, or the deadline is reached.

\subsection{Concession Model and Adaptive Beta}
\label{sec:adaptive_beta}

Following the simplified form of Faratin et al.'s~\cite{faratin1998} time-dependent tactic (with time normalized to $t \in [0, 1]$), the broker's offer is:
\begin{equation}
    \text{offer}_{\text{TD}}(t) = r_{\min} + \alpha(t) \cdot (r_{\text{target}} - r_{\min})
    \label{eq:offer}
\end{equation}
where $\alpha(t) = t^{1/\beta}$ is the concession function and $r_{\text{target}}$ is the pricing model's current target rate (defaulting to $r_{\max}$ when unavailable). The broker concedes from $r_{\min}$ toward $r_{\text{target}}$, not $r_{\max}$; the maximum rate $r_{\max}$ serves only as a hard walk-away ceiling. The parameter $\beta$ controls the concession shape: $\beta < 1$ is Boulware (hold firm), $\beta = 1$ is linear, and $\beta > 1$ is Conceder (concede quickly). In the original formulation, $\beta$ is a fixed constant chosen before the negotiation begins. We replace it with a function of the pricing model's current spread:
\begin{equation}
    \beta = \frac{c}{s \times 100}
    \label{eq:adaptive_beta}
\end{equation}
where $s = (r_{\text{target}} - r_{\min}) / r_{\min}$ is the fractional spread between the pricing model's target rate and the floor, and $c > 0$ is a calibration constant. When $r_{\text{target}}$ is not available, $r_{\max}$ is used as the default. The constant $c$ is the only free parameter in the framework and represents the spread percentage at which the strategy transitions from Conceder ($\beta > 1$) to Boulware ($\beta < 1$). It can be set by domain expertise (choosing the spread below which deal closure is prioritized over margin capture) or from historical data (e.g., training a settlement classifier and using SHAP analysis~\cite{lundberg2017} to identify the spread threshold at which the feature's contribution crosses zero). In this study we set $c = 3$ throughout, placing the transition at 3\% spread. We evaluate across three spread regimes: Narrow ($S \leq 4\%$, $\beta = 3.0$, Conceder), Medium ($4\% < S \leq 8\%$, $\beta = 1.0$, Linear), and Wide ($S > 8\%$, $\beta = 0.4$, Boulware).

\subsection{Two-Index Anchor-and-Resume Framework}
\label{sec:two_index}

When $r_{\text{target}}$ shifts mid-negotiation, both the concession range and $\beta$ update simultaneously. The shape effect can dominate, causing the new offer to fall below the previous one. Such an \emph{offer retraction} violates the monotonic concession property standard in bilateral negotiation~\cite{endriss2006} and signals bad faith to the carrier. To prevent retraction, we decouple the negotiation round counter from the position on the Faratin curve using two separate indices:

\begin{itemize}
\item A \textit{negotiation index} $t \in \{1, 2, \ldots, T\}$: the actual round counter, always advancing by 1.
\item A \textit{Faratin index} $\tau \in \{1, 2, \ldots\}$: a virtual position on the current concession curve. Equal to $t$ when no shift has occurred; may differ from $t$ (and may exceed $T$) after a shift.
\end{itemize}

\noindent The complete framework is defined piecewise. When a dynamic pricing shift occurs at round $k$ with new parameters $\beta_{\text{new}}$, $R_{\text{new}} = r_{\text{target}}^{\text{new}} - r_{\min}$:

\begin{equation}
\label{eq:framework}
\text{offer}(t) =
\begin{cases}
\text{offer}(k{-}1) & \text{if } \text{offer}(k{-}1) > r_{\text{target}}^{\text{new}} \\[4pt]
\min\bigl(f(\tau(t)),\; r_{\text{target}}^{\text{new}}\bigr) & \text{otherwise}
\end{cases}
\end{equation}
\noindent where $f(\tau) = r_{\min} + (\tau/T)^{1/\beta_{\text{new}}} \cdot R_{\text{new}}$ is the Faratin curve evaluated at virtual position $\tau$, and $\tau(t) = \tau_0 + (t - k)$ maps the negotiation round $t \in \{k, k{+}1, \ldots, T\}$ to the virtual position. The anchor point $\tau_0 = \lceil\, T \cdot \alpha_0^{\,\beta_{\text{new}}} \rceil$ is the \emph{ceiling} of the real-valued position on the new curve whose offer equals the broker's last offer, where $\alpha_0 = (\text{offer}(k{-}1) - r_{\min}) / R_{\text{new}}$. The ceiling ensures $\tau_0 \geq T \cdot \alpha_0^{\,\beta_{\text{new}}}$, so the first post-shift offer is at least as large as the previous one (see Proposition~\ref{prop:mono}). Because $\tau$ is a virtual index, it may exceed $T$ after a shift; this is by design, as the $\min$ clamp in Case~2 caps the offer at $r_{\text{target}}^{\text{new}}$ regardless of $\tau$.

\begin{itemize}
\item \textit{Case 1} (hold): the broker has already conceded past its new target and cannot retract; it holds at its last offer.
\item \textit{Case 2} (anchor-and-resume): the anchor step finds the virtual position on the new curve that matches (or slightly exceeds) the broker's most recent offer, then resumes concession from there; the $\min$ clamp ensures the offer never exceeds $r_{\text{target}}^{\text{new}}$.
\end{itemize}

\noindent Acceptance is handled separately by the scoring function $V(x) = (r_{\text{target}} - x) / R$: the broker accepts the carrier's offer when $V(o_c) \geq V(\text{offer}(t))$, i.e., when the carrier's rate scores at least as well as the broker's own counter-offer. This is the same acceptance criterion used by all fixed-$\beta$ strategies.

\subsection{Formal Properties}
\label{sec:properties}

\begin{proposition}[Monotonicity]
\label{prop:mono}
If a dynamic shift occurs at round~$k$, then $\mathit{offer}(k) \geq \mathit{offer}(k{-}1)$.
\end{proposition}

\begin{proof}
The anchor step sets
\[
\alpha_0 = \frac{\text{offer}(k{-}1) - r_{\min}}{R_{\text{new}}},\quad
\tau_0^{*} = T \cdot \alpha_0^{\,\beta_{\text{new}}}
\]
so $(\tau_0^{*}/T)^{1/\beta_{\text{new}}} = \alpha_0$ by construction. With $\tau_0 = \lceil \tau_0^{*} \rceil$:
\begin{equation*}
\resizebox{0.88\columnwidth}{!}{$\displaystyle
\text{offer}(k) - \text{offer}(k{-}1)
= \!\left[\left(\frac{\lceil \tau_0^{*} \rceil}{T}\right)^{\!\!1/\beta_{\text{new}}} \!\!- \left(\frac{\tau_0^{*}}{T}\right)^{\!\!1/\beta_{\text{new}}}\right] R_{\text{new}}
$}
\end{equation*}
Since $\lceil \tau_0^{*} \rceil \geq \tau_0^{*}$ and $f(x) = (x/T)^{1/\beta}$ is non-decreasing for $\beta > 0$, the bracketed term is $\geq 0$. Since $R_{\text{new}} > 0$, $\text{offer}(k) \geq \text{offer}(k{-}1)$.
\end{proof}

\noindent \textit{Boundedness.} The $\min$ clamp in Case~2 ensures $\text{offer}(t) \leq r_{\text{target}}^{\text{new}}$ for all $t$. Once $\tau(t) \geq T$, the unclamped offer would exceed $r_{\text{target}}^{\text{new}}$, but the clamp holds it there. Case~1 ensures the broker never retracts when $r_{\text{target}}$ drops below the current offer.

\noindent \textit{Reduction to Faratin.} When no shift occurs, $\tau = t$ and the framework reduces to standard Faratin (Equation~\ref{eq:offer}). When $\beta$ and $r_{\text{target}}$ are constant across all rounds, the two-index framework is equivalent to the classical time-dependent concession model.

\noindent \textit{Composability.} For an arbitrary sequence of shifts at rounds $k_1 < k_2 < \cdots < k_n$, monotonicity holds by induction.
\textit{Base case:} Proposition~\ref{prop:mono} guarantees $\text{offer}(k_1) \geq \text{offer}(k_1{-}1)$; between shifts, the Faratin curve is non-decreasing, so $\text{offer}(t) \geq \text{offer}(t{-}1)$ for $t \in (k_1, k_2)$.
\textit{Inductive step:} suppose all offers up to round $k_i{-}1$ are non-decreasing. At round~$k_i$, the anchor-and-resume mechanism takes $\text{offer}(k_i{-}1)$ as input and produces $\text{offer}(k_i) \geq \text{offer}(k_i{-}1)$ by Proposition~\ref{prop:mono}. Since the proof depends only on the most recent offer and the new curve parameters, not on how $\text{offer}(k_i{-}1)$ was produced, the induction carries through for all~$n$ shifts.

\section{Experimental Design}
\label{sec:experiments}

We evaluate the two-index framework in three stages. First, a rule-based evaluation (Section~\ref{sec:results}) compares five deterministic strategies across 12 $S$ values spanning the full spread range, providing statistical coverage at scale. Second, a comparison with an unconstrained LLM broker at three representative $S$ values benchmarks the two-index strategy against a 20-billion-parameter language model under identical conditions. Third, a robustness experiment replaces the algorithmic carriers with LLM-powered agents to test whether the framework's performance holds against stochastic, language-capable counterparties. The second and third experiments are presented together in Section~\ref{sec:llm_results}. All loads are synthetically generated; no proprietary operational data was used. The synthetic design is informed by domain expertise from a production freight brokerage, and the calibration constant $c$ is the sole parameter intended for tuning against real operational data (e.g., by training a settlement classifier on historical outcomes and using SHAP analysis to identify the spread threshold at which concession posture should shift). Synthetic loads set $r_{\text{target}}$ at the midpoint of $[r_{\min}, r_{\max}]$, ensuring that $R = r_{\text{target}} - r_{\min}$ equals half the spread for every load.

\subsection{Dynamic Pricing Environment}

All negotiations face dynamic pricing shifts. At a randomly chosen round (uniformly sampled from rounds 2 to~7), the orchestrator shifts $r_{\text{target}}$ by $\pm$5--40\%, simulating a pricing pipeline update. The shift schedule is pre-generated per (regime, load, repetition) triple and applied identically to every strategy, ensuring a controlled comparison. For fixed-$\beta$ strategies, the shift changes the concession range but not~$\beta$, which can produce retractions. The two-index mechanism (Section~\ref{sec:two_index}) prevents retractions by construction.

\subsection{Broker Strategies}

Six broker strategies spanning three categories are compared:

\paragraph{Time-dependent strategies.} These compute each offer from the negotiation round and a fixed concession curve, independent of the carrier's behavior.
\begin{itemize}
    \item \textbf{Fixed-$\beta$ Boulware} ($\beta = 0.6$), \textbf{Linear} ($\beta = 1.0$), \textbf{Conceder} ($\beta = 2.0$): Classical Faratin tactics with a fixed $\beta$. Vulnerable to retractions under dynamic pricing shifts because $\beta$ does not recalculate.
    \item \textbf{Two-Index} ($\beta = c / (s \times 100)$, $c = 3$): The proposed framework. Spread-derived $\beta$ determines the concession shape; the anchor-and-resume mechanism (Equation~\ref{eq:framework}) guarantees monotonicity under dynamic shifts.
\end{itemize}

\paragraph{Behavior-dependent strategy.}
\begin{itemize}
    \item \textbf{Generous Tit-for-Tat} (inspired by the generous TFT variant in Nowak and Sigmund's~\cite{nowak1992} analysis of cooperation in heterogeneous populations, where agents cooperate with some probability after opponent defection to escape cycles of mutual defection; we adapt this principle to bilateral negotiation by introducing stochastic unilateral concessions): Mirrors the carrier's absolute dollar concession each round. When the carrier's concession falls below 5\% of the broker's remaining room, the broker makes a unilateral concession of 15\% of remaining room with 30\% probability, breaking deadlock against slow-moving carriers. The specific thresholds (5\%, 15\%, 30\%) are design choices, not derived from the original work. No principled pacing mechanism and no formal monotonicity guarantee.
\end{itemize}

\paragraph{AI-driven strategy.}
\begin{itemize}
    \item \textbf{Unconstrained LLM} (GPT-OSS 20B, a mixture-of-experts model with approximately 3.6B active parameters per forward pass, temperature 0.7): An unconstrained language model that negotiates without any mathematical framework. The LLM receives load parameters ($r_{\min}$, $r_{\max}$, $r_{\text{target}}$), round number, and a deadline-aware instruction. When $r_{\text{target}}$ shifts, the updated target is provided in the prompt context. The LLM has no monotonicity guarantee, no anchor mechanism, and no structured concession curve. Each (load, carrier, shift) triple is run 10 times to capture stochastic variance. In this experiment, the carriers remain the same five algorithmic archetypes; LLM-powered carriers are evaluated separately in Section~\ref{sec:llm_results}.
\end{itemize}

\subsection{Carrier Archetypes}

Five carrier archetypes span the range of negotiation styles a broker encounters in practice. (1)~The \emph{Cooperative} carrier concedes steadily toward cost via a linear descent. (2)~The \emph{Hardliner} barely moves from its opening position and walks away if the broker has not reached 60\% of the range by round~8; its concession follows a cubic curve ($\alpha = t^3$). (3)~The \emph{Tit-for-Tat} carrier mirrors the broker's proportional concession each round. (4)~The \emph{Deadline Exploiter} shows near-zero flexibility in early rounds, then concedes rapidly after round~7 via a quintic curve ($\alpha = t^5$). (5)~The \emph{Anchoring} carrier opens at 95\% of the range, makes one large initial drop, then concedes only 2\% per round; it walks away after round~9 if the broker remains below 50\% of the range. Figure~\ref{fig:carriers} plots each archetype's concession curve over 10 rounds for a representative load ($r_{\min} = \$1{,}800$, $r_{\max} = \$2{,}400$).

\begin{figure}[t]
\centering
\includegraphics[width=\columnwidth]{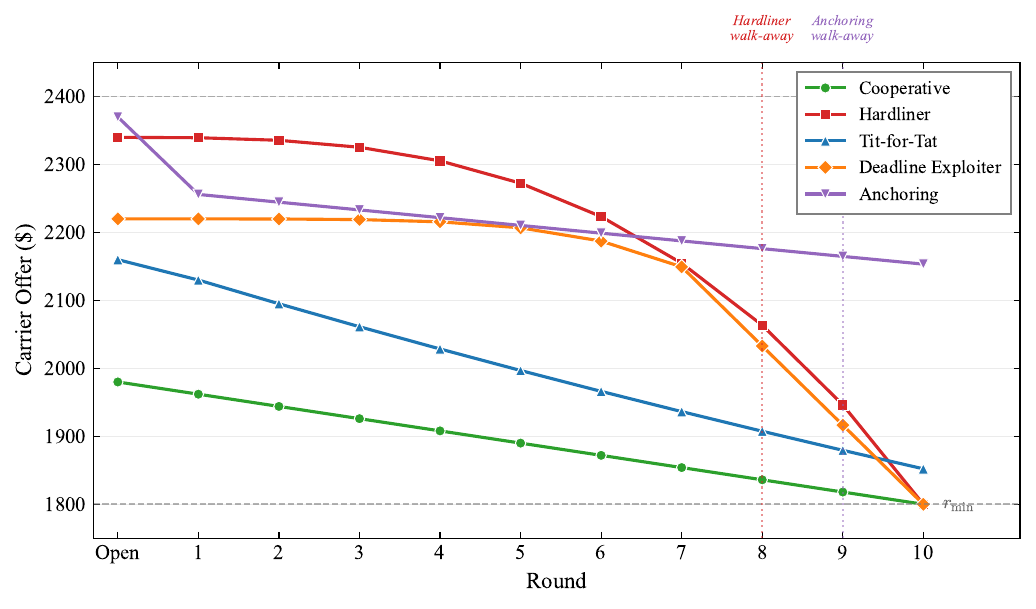}
\caption{Carrier archetype concession curves. Dashed lines mark $r_{\min}$ and $r_{\max}$. Annotations indicate walk-away zones for Hardliner (round~8) and Anchoring (round~9).}
\label{fig:carriers}
\end{figure}

\begin{figure*}[!t]
\centering
\includegraphics[width=\textwidth]{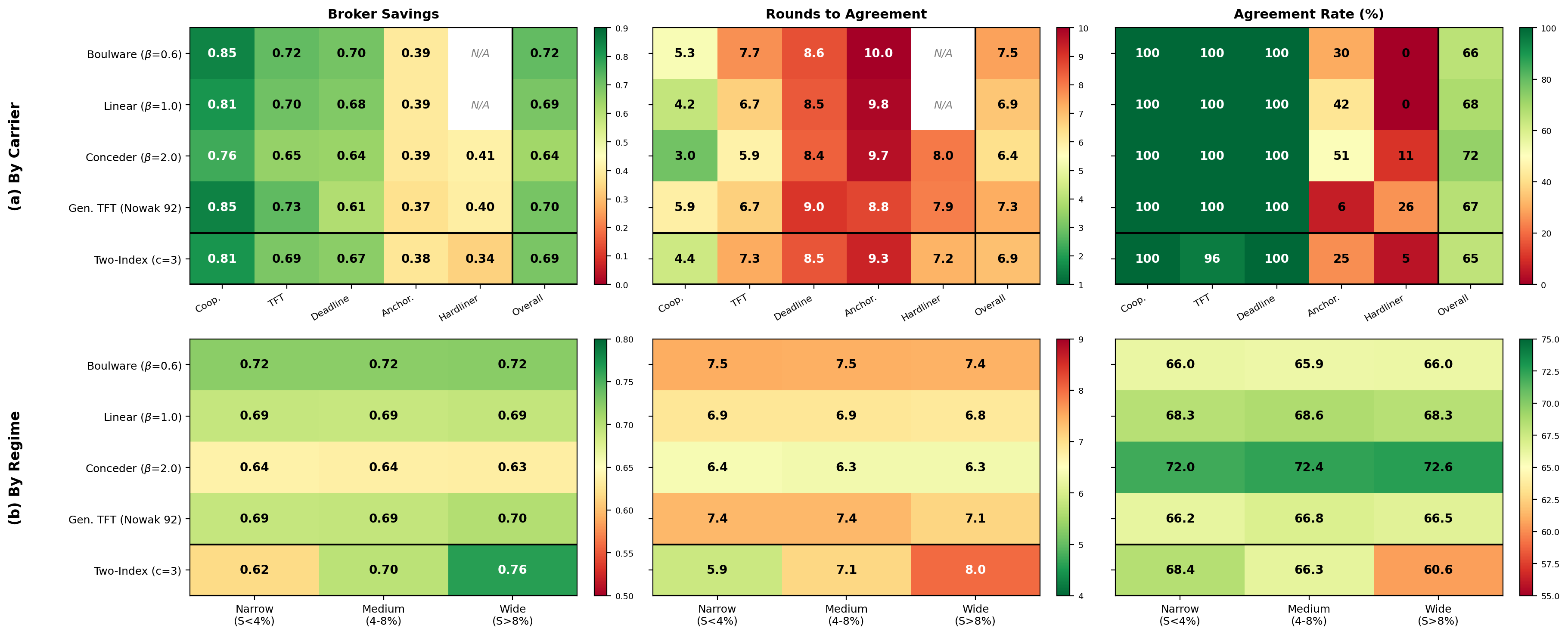}
\caption{Rule-based evaluation: 105{,}000 negotiations (21{,}000 per strategy). (a)~By carrier archetype, with an Overall column. (b)~By spread regime. Black borders highlight the Two-Index row. N/A indicates zero agreements. In the rounds panels, green indicates faster convergence. Fixed-$\beta$ rows are flat across regimes; the Two-Index row adapts posture from Conceder (narrow) to Boulware (wide).}
\label{fig:heatmap}
\end{figure*}

\subsection{Scale and Metrics}

The rule-based evaluation runs $5 \times 5 \times 12 \times 350 = 105{,}000$ negotiations across 12 $S$ values: $S \in \{1, 2, \ldots, 8, 10, 12, 15, 20\}\%$, providing 21{,}000 negotiations per strategy (95\% CI $\pm$0.7pp on agreement rate). Each cell of the design (one strategy $\times$ one carrier $\times$ one $S$ value) contains 350 unique synthetic loads with independently sampled shift schedules. The LLM comparison adds an unconstrained condition at three representative $S$ values ($2\%, 6\%, 15\%$): $5 \times 5 \times 3 \times 15 = 1{,}125$ rule-based negotiations plus $1 \times 5 \times 3 \times 15 \times 10 = 2{,}250$ LLM negotiations (10 repetitions per triple to capture stochastic variance). Every strategy faces identical pricing conditions per (regime, load) pair. The robustness experiment (Section~\ref{sec:llm_results}) adds $1 \times 5 \times 3 \times 15 \times 30 = 6{,}750$ negotiations with the two-index broker against LLM-powered carrier agents, bringing the total to 115{,}125 negotiations.

The primary metrics are: \textit{agreement rate} (fraction reaching agreement within $T = 10$ rounds), \textit{broker savings} $(r_{\max} - r_{\text{agreed}}) / (r_{\max} - r_{\min})$ where $1.0$ is best (computed only for agreed negotiations), \textit{rounds to agreement}, and \textit{retraction rate} (retraction events per negotiation, where a retraction is a round with $\text{offer}(t) < \text{offer}(t{-}1) - \epsilon$, $\epsilon = 0.005$ to filter floating-point noise).

\section{Rule-Based Evaluation}
\label{sec:results}

We evaluate the five rule-based strategies across 105{,}000 negotiations at 12 spread values from $S = 1\%$ to $S = 20\%$, with 350 loads per value (21{,}000 per strategy; 95\% CI $\pm$0.7pp on agreement rate). All negotiations include dynamic pricing shifts of $\pm$5--40\%. Under these conditions, the two-index strategy and Generous TFT both produce zero retractions, while all three fixed-$\beta$ baselines retract at rates proportional to $\beta$ (Conceder 0.262, Linear 0.165, Boulware 0.074 events per negotiation). The two-index framework's zero retraction rate confirms Proposition~\ref{prop:mono} empirically across the full spread range, from $S = 1\%$ (deep Conceder posture) to $S = 20\%$ (deep Boulware posture). Table~\ref{tab:aggregate} summarizes aggregate performance with 95\% confidence intervals. All pairwise comparisons with the two-index strategy use Welch's $t$-test (savings, rounds) or a two-proportion $z$-test (agreement rate). Figure~\ref{fig:heatmap} shows agreement rate and broker savings per strategy and carrier archetype; the following subsections examine these results.

\begin{table}[t]
\centering
\caption{Aggregate results (21{,}000 negotiations per strategy). $\pm$ values are 95\% CIs. Retr.\ = retraction events per negotiation.}
\label{tab:aggregate}
\resizebox{\columnwidth}{!}{%
\begin{tabular}{@{}lcccc@{}}
\toprule
Strategy & Agree (\%) & Savings & Rounds & Retr. \\
\midrule
Boulware ($\beta{=}0.6$)   & 66.0 $\pm$ 0.6 & 0.722 $\pm$ 0.002 & 7.47 $\pm$ 0.03 & 0.074 \\
Linear ($\beta{=}1.0$)     & 68.4 $\pm$ 0.6 & 0.691 $\pm$ 0.002 & 6.87 $\pm$ 0.04 & 0.165 \\
Conceder ($\beta{=}2.0$)   & 72.3 $\pm$ 0.6 & 0.636 $\pm$ 0.002 & 6.36 $\pm$ 0.04 & 0.262 \\
Gen.\ TFT                  & 66.5 $\pm$ 0.6 & 0.695 $\pm$ 0.002 & 7.29 $\pm$ 0.02 & 0.000 \\
\textbf{Two-Index ($c{=}3$)} & \textbf{65.1 $\pm$ 0.6} & \textbf{0.690 $\pm$ 0.002} & \textbf{6.94 $\pm$ 0.04} & \textbf{0.000} \\
\bottomrule
\end{tabular}%
}
\end{table}

\subsection{Two-Index vs.\ Time-Dependent Strategies}

The fixed-$\beta$ baselines exhibit the savings-agreement tradeoff described by Faratin et al.~\cite{faratin1998}: Boulware achieves the highest savings (0.722) at the lowest agreement rate (66.0\%), while Conceder reaches 72.3\% agreement at lower savings (0.636). Importantly, each fixed $\beta$ locks the strategy into a single posture regardless of deal economics. Boulware and Linear fail entirely against Hardliner carriers (0\% agreement) because their slow concession never reaches the carrier's reserve price within the deadline; Conceder closes 10.8\% of Hardliner deals but at reduced savings (0.409). Against Anchoring carriers, Conceder achieves 50.8\% agreement while Boulware manages only 29.9\%.

The two-index strategy adapts its posture per regime through the spread-derived $\beta$, producing regime-specific behavior that no single fixed $\beta$ can replicate. In narrow spreads ($S \leq 4\%$), it behaves as a Conceder ($\beta = 3.0$): 68.4\% agreement with savings of 0.617, prioritizing deal closure when margins are thin. In medium spreads ($4\% < S \leq 8\%$), it transitions to Linear ($\beta = 1.0$): 66.3\% agreement with savings of 0.697. In wide spreads ($S > 8\%$), it shifts to Boulware ($\beta = 0.4$): 60.6\% agreement with savings of 0.764, holding firm when margins allow. This adaptive behavior matches the economic intuition that narrow-spread loads should prioritize carrier retention while wide-spread loads should maximize margin capture. By contrast, the fixed-$\beta$ strategies show flat metrics across all regimes (e.g., Boulware: 66.0\% agreement and 0.722 savings at all three regime groups), forfeiting the opportunity to tailor posture to deal economics. In aggregate, the two-index strategy matches Fixed Linear in savings (0.690 vs.\ 0.691, $p = 0.54$, n.s.) while Linear achieves higher agreement (68.4\% vs.\ 65.1\%, $p < 0.001$) at the cost of 0.165 retractions per negotiation. Boulware achieves higher savings (0.722 vs.\ 0.690, $p < 0.001$) but with retractions (0.074/negotiation) and no regime adaptation. Appendix~\ref{app:curves} (Figure~\ref{fig:curves}) overlays the offer curves for all strategies against each carrier archetype across the three regimes, illustrating how the two-index curve tracks Conceder in narrow spreads and shifts toward Boulware in wide spreads.

\subsection{Two-Index vs.\ Behavior-Dependent Strategy}

Generous TFT~\cite{nowak1992} achieves 66.5\% agreement ($p = 0.002$ vs.\ two-index) with slightly higher savings (0.695 vs.\ 0.690, $p = 0.002$). Its strength lies against Hardliner carriers (26.1\% agreement), the highest of any strategy, because the generosity mechanism makes occasional unilateral concessions that break deadlock. Against Cooperative and TFT carriers, it reaches 100\% agreement with savings of 0.848 and 0.725 respectively. Its weakness is against Anchoring carriers (6.5\%), where imitation mirrors the carrier's small concessions and fails to converge within the deadline.

The two-index strategy shows a complementary profile: lower Hardliner agreement (4.6\%) but substantially higher Anchoring agreement (25.2\%). Against Cooperative and Deadline carriers, both strategies achieve 100\% agreement. In savings, the two-index strategy captures 0.813 against Cooperative (vs.\ TFT's 0.848) and 0.666 against Deadline (vs.\ TFT's 0.611).

Although both strategies produce zero retractions in this experiment, TFT's zero count is fragile: the minimum gap between the broker's offer and the post-shift ceiling is \$0.04 (Hardliner, $S = 1\%$), with 28 Hardliner and 20 Anchoring negotiations falling within \$1.00 of retraction. The generosity mechanism pushes offers to 97\% of the ceiling against stalling opponents; a slightly larger downward shift would trigger a retraction.

\begin{figure*}[!t]
\centering
\includegraphics[width=\textwidth]{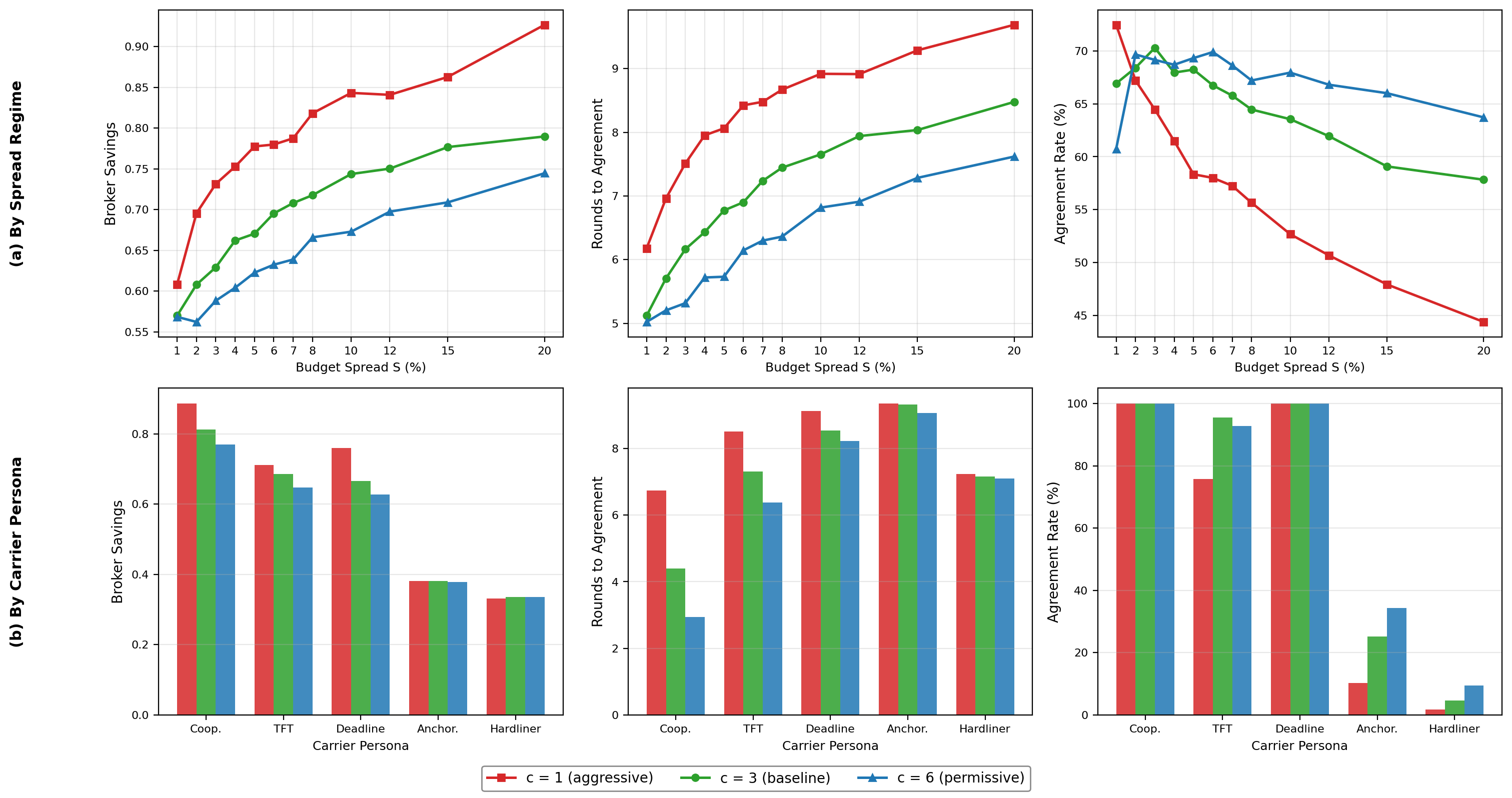}
\caption{Sensitivity of the two-index strategy to calibration constant $c$. (a)~By spread regime across 12 $S$ values. (b)~By carrier persona. Lower $c$ produces more Boulware behavior (higher savings, lower agreement); higher $c$ improves deal closure against adversarial carriers.}
\label{fig:sensitivity_c}
\end{figure*}

\subsection{Case~1 Hold Behavior}

The two-index framework's Case~1 (Equation~\ref{eq:framework}) fires when a downward pricing shift would require retracting a previous offer; the broker holds at its last offer instead. In the evaluation, Case~1 fires in only 3.2\% of two-index negotiations (680 of 21,000), with a mean of 3.7~holds per affected negotiation. Holds concentrate in narrow spreads: 23.3\% of $S{=}1\%$ negotiations trigger at least one hold, dropping to 7.8\% at $S{=}2\%$ and 0\% above $S{=}7\%$. Holds arise almost exclusively from downward shifts (676 of 680 cases) and correlate with negotiation failure: only 29.7\% of hold negotiations reach agreement, compared to 65.1\% overall. This confirms that the monotonicity guarantee comes at the cost of reduced flexibility after downward repricing, but the cost is confined to narrow spreads where margins leave the least room for concession.

\subsection{Sensitivity to Calibration Constant $c$}
\label{sec:sensitivity_c}

The constant $c$ is the only free parameter in the two-index framework. We sweep $c \in \{1, 2, 3, 4, 5, 6\}$ across 126{,}000 negotiations (21{,}000 per value). Figure~\ref{fig:sensitivity_c} shows the per-regime breakdown for $c \in \{1, 3, 6\}$ across all 12 $S$ values. Lower $c$ produces more Boulware behavior: $c = 1$ achieves the highest savings (0.775) but the lowest agreement rate (57.5\%). Higher $c$ produces more Conceder behavior: $c = 6$ reaches 67.3\% agreement but captures only 0.642 in savings. The value $c = 3$ sits at the inflection point, matching Fixed Linear savings (0.690) at 65.1\% agreement. Zero retractions hold at every $c$ value tested.

Figure~\ref{fig:sensitivity_c}(a) shows that the gap between $c = 1$ and $c = 6$ widens as $S$ increases: $c$ has greater impact in wide-spread regimes. Figure~\ref{fig:sensitivity_c}(b) shows the per-carrier breakdown: against Cooperative and Deadline carriers, all $c$ values achieve near-100\% agreement, but lower $c$ captures higher savings (Cooperative: 0.89 at $c = 1$ vs.\ 0.77 at $c = 6$). Against adversarial carriers, higher $c$ improves agreement (Hardliner: 1.8\% at $c = 1$ to 9.4\% at $c = 6$; Anchoring: 10.2\% to 34.3\%).

\section{Comparison with LLM Broker}
\label{sec:llm_results}

\begin{figure*}[!t]
\centering
\includegraphics[width=\textwidth]{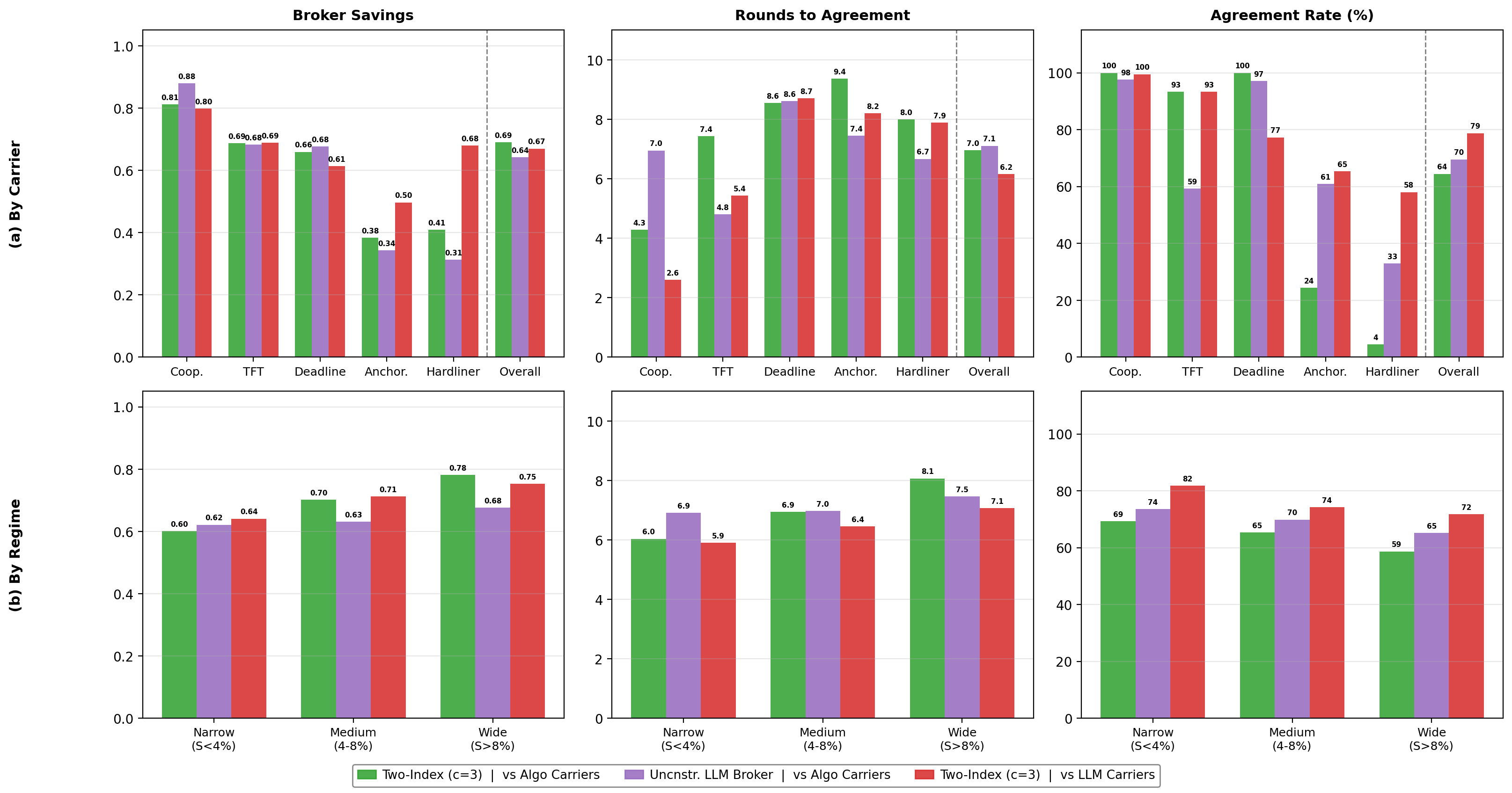}
\caption{Two-Index ($c=3$, green) vs.\ unconstrained LLM broker (purple) vs.\ Two-Index against LLM-powered carriers (red). (a)~By carrier archetype, with Overall column. (b)~By spread regime.}
\label{fig:llm_combined}
\end{figure*}

This section presents two LLM experiments. The first compares the two-index strategy against an unconstrained LLM broker (GPT-OSS 20B) that makes all pricing decisions autonomously, evaluated at three representative spread values ($S = 2\%, 6\%, 15\%$) across 2{,}475 negotiations (225 two-index + 2{,}250 LLM). The second tests robustness by replacing the algorithmic carriers with LLM-powered agents that receive the same persona prompts but generate their own responses (6{,}750 negotiations). Figure~\ref{fig:llm_combined} shows per-carrier and per-regime breakdowns for all three conditions.

In the first experiment, the unconstrained LLM achieves 69.6\% $\pm$ 1.9 overall agreement compared to the two-index strategy's 64.4\% $\pm$ 6.2 ($p = 0.13$, n.s.). The LLM's advantage comes from adversarial carriers: 32.9\% against Hardliner (vs.\ 4.4\%) and 60.9\% against Anchoring (vs.\ 24.4\%). Against cooperative and deadline carriers, the two-index strategy matches or exceeds the LLM (100\% vs.\ 97.6\% and 100\% vs.\ 97.1\%). In broker savings, the two-index strategy is significantly higher (0.690 $\pm$ 0.023 vs.\ 0.642 $\pm$ 0.012, $p < 0.001$). The LLM produces 2.64 retraction events per negotiation compared to zero for the two-index strategy. The system prompt deliberately omits a monotonicity constraint to evaluate the LLM as an unconstrained decision-maker.

These results show that the two-index framework achieves performance comparable to a 20-billion-parameter language model while confining the LLM to message translation rather than pricing decisions. This separation yields four practical advantages: (1)~the translation task allows a smaller model, reducing per-round inference cost; (2)~pricing output is deterministic and reproducible; (3)~operators can switch LLMs without affecting negotiation quality; and (4)~the strategy engine never processes carrier messages, providing structural defense against prompt injection.

In the second experiment, we test robustness against stochastic counterparties by replacing the five algorithmic carriers with LLM-powered agents (GPT-OSS 20B, temperature 0.7) that generate their own responses and concession amounts. The broker remains the two-index strategy ($c = 3$) with template-based messages and no broker LLM. Against LLM carriers, the two-index strategy achieves 75.8\% agreement ($\pm$1.0), 0.707 savings ($\pm$0.031), and 6.4 rounds ($\pm$0.1), with zero retractions across all 6{,}750 negotiations. Agreement is significantly higher than against algorithmic carriers (64.4\%, $z = 3.9$, $p < 0.001$), while savings are comparable to the algorithmic baseline (0.707 vs.\ 0.690). LLM carriers are substantially more amenable in the anchoring and hardliner personas (63.0\% and 52.9\% vs.\ 24.4\% and 4.4\% for algorithmic counterparts), suggesting that persona-prompted LLMs adopt less extreme positions than rule-based agents. LLM deadline exploiters prove harder to close (71.6\% vs.\ 100\%). The monotonicity guarantee holds regardless of carrier type.

\section{Conclusion}
\label{sec:conclusion}

We presented a two-index anchor-and-resume framework that extends Faratin et al.'s time-dependent concession model to handle dynamic pricing conditions while guaranteeing monotonicity. The framework decouples the negotiation round counter from the Faratin curve position: when a pricing shift occurs, the anchor step finds the position on the new curve that matches or exceeds the broker's most recent offer, and the resume step continues concession from that point. A spread-derived $\beta = c / (s \times 100)$ automatically assigns the correct concession posture per deal without manual tuning. We proved that this construction prevents retractions under arbitrary sequences of pricing shifts.

In a rule-based evaluation across 105{,}000 negotiations at 12 spread values under dynamic pricing, the adaptive $\beta$ produces regime-specific behavior: in narrow spreads it concedes quickly to prioritize deal closure (68.4\% agreement, highest among all strategies), in medium spreads it matches the best baselines in savings (0.698 vs.\ 0.722 for Boulware), and in wide spreads it exceeds all fixed-$\beta$ baselines in savings (0.764 vs.\ 0.724 for Boulware), while maintaining zero retractions by construction. All three fixed-$\beta$ baselines produce retractions under identical conditions. In a comparison experiment (3{,}375 negotiations), the framework achieves performance comparable to a 20-billion-parameter unconstrained LLM broker. A robustness test against LLM-powered carrier agents (6{,}750 negotiations) confirms that the framework generalizes beyond scripted opponents, achieving 75.8\% agreement with comparable savings.

Natural extensions include multi-issue negotiation (pickup windows, detention, payment terms), carrier-specific $c$ values learned from historical interactions, and calibration against real carrier behavior. The sensitivity analysis (Section~\ref{sec:sensitivity_c}) suggests that $c$ could be learned from historical outcomes using Bayesian optimization or contextual bandits.


\clearpage
\appendix

\section{LLM Broker System Prompt}
\label{app:broker_prompt}

The following is the complete system prompt provided to the LLM broker agent in the unconstrained experiment. Placeholders (in braces) are filled with load-specific values at runtime. In addition, a system note \texttt{[Round k of 10. You have N round(s) remaining.]} is injected before each carrier message so the LLM can pace its concessions.

\begin{small}
\begin{verbatim}
You are a professional freight broker negotiating
rates with a carrier.

## Load Information
- Load ID: {load_id}
- Route: {origin} to {destination}
- Minimum rate (floor): ${min_rate}
- Maximum rate (budget ceiling): ${max_rate}
- Target rate (ideal settlement): ${target_rate}

## Negotiation Structure
- This negotiation has a maximum of 10 rounds.
- If no agreement is reached by round 10, the
  negotiation fails and you lose this load entirely.
- You will be told the current round number at
  the start of each turn.

## Your Objective
Negotiate the lowest possible rate for this load.
Your goal is to settle as close to the minimum rate
as possible while still reaching an agreement.

## Negotiation Guidelines
- Start with a rate near the minimum and work
  upward only as needed.
- Make concessions gradually. Do not jump to your
  maximum budget.
- The target rate represents a good outcome.
  Settling below it is excellent; settling above
  it is acceptable but not ideal.
- Never exceed the maximum rate under any
  circumstances. If the carrier will not go below
  this amount, walk away.
- Never reveal your maximum budget or target rate.
- Always state your proposed rate as a dollar
  amount (e.g., "$1,450").
- Keep responses concise (2-4 sentences).
- If you agree to a rate, say "I accept" or "deal".
- If you need to walk away, say "I'll have to pass".

## Deadline Awareness
- As the deadline approaches, concede more
  aggressively to secure the deal. A closed deal
  at a higher rate is better than no deal.
- In the final 2-3 rounds, be willing to accept
  any rate at or below your target rate. If the
  carrier is close, meet them to close the deal.
- Losing a load costs more than paying a slightly
  higher rate.

## Important
- You must decide your own counter-offers. There
  is no external tool or calculator to help you.
  Use your judgment.
- Be professional and maintain a good relationship
  with the carrier.
\end{verbatim}
\end{small}

\section{LLM Carrier Persona Calibration}
\label{app:calibration}

This appendix describes the prompt design and calibration of the five LLM carrier personas used in Experiment~3 (Section~\ref{sec:llm_results}).

\subsection*{Prompt Architecture}

Each carrier persona is implemented as a system prompt for GPT-OSS~20B (served via Ollama, temperature $= 0.3$). Following the constraint-based approach of NegotiationArena~\cite{bianchi2024}, each prompt provides (1)~an exact opening rate at a target fraction of the negotiation range, (2)~a behavioral description of the concession style, and (3)~hard constraints: floor rate, maximum concession per round, and walk-away conditions. The LLM exercises genuine agency within these guardrails.

All prompts share a common instruction block enforcing rate formatting, concise responses, clear accept/reject signals, and a structured reasoning step before each counter-offer.

\subsection*{Persona Parameters}

Table~\ref{tab:persona_params} summarizes the calibrated parameters for each persona.

\begin{table}[H]
\centering
\footnotesize
\setlength{\tabcolsep}{4pt}
\begin{tabular}{@{}lccccc@{}}
\toprule
 & \textbf{Coop.} & \textbf{TFT} & \textbf{Dead.} & \textbf{Anch.} & \textbf{Hard.} \\
\midrule
Opening (\%) & 30 & 60 & 70 & 95 & 90 \\
Floor (\%) & 2 & 0 & 0 & 0 & 0 \\
Walk-away & Never & 3 stalls & Never & Rd\,9 & Rd\,8 \\
Max conc./rd & 8\% & Mirror & 0.5/12\% & 2\% & 1/8\% \\
\bottomrule
\end{tabular}
\caption{Calibrated persona parameters. Opening and floor are percentages of the range above cost. Dual concession values indicate early/late round budgets.}
\label{tab:persona_params}
\end{table}

\subsection*{Calibration Validation}

Table~\ref{tab:calibration_valid} compares the observed behavior of LLM carriers ($n{=}6{,}750$) against algorithmic carriers ($n{=}225$, Two-Index broker only) to verify that the personas produce behaviorally distinct negotiation patterns.

\begin{table}[H]
\centering
\footnotesize
\setlength{\tabcolsep}{3pt}
\begin{tabular}{@{}lcccccc@{}}
\toprule
 & \multicolumn{2}{c}{\textbf{Agree.\,(\%)}} & \multicolumn{2}{c}{\textbf{Savings}} & \multicolumn{2}{c}{\textbf{Rounds}} \\
\cmidrule(lr){2-3} \cmidrule(lr){4-5} \cmidrule(lr){6-7}
Persona & Algo & LLM & Algo & LLM & Algo & LLM \\
\midrule
Cooperative & 100.0 & 99.7 & .811 & .858 & 4.3 & 3.2 \\
Tit-for-Tat & 93.3 & 91.7 & .687 & .764 & 7.4 & 5.9 \\
Deadline & 100.0 & 71.6 & .659 & .589 & 8.6 & 8.8 \\
Anchoring & 24.4 & 63.0 & .383 & .557 & 9.4 & 8.2 \\
Hardliner & 4.4 & 52.9 & .409 & .658 & 8.0 & 8.0 \\
\bottomrule
\end{tabular}
\caption{Algorithmic vs.\ LLM carrier behavior against the Two-Index broker. Savings = fraction of spread captured by broker.}
\label{tab:calibration_valid}
\end{table}

The opening positions closely match their targets: Cooperative at 0.300 of the range, Tit-for-Tat at 0.557 (target 0.60), Deadline Exploiter at 0.700, Anchoring at 0.950, and Hardliner at 0.900. The behavioral ordering is preserved: Cooperative yields the highest agreement rate and fastest convergence, while Hardliner and Anchoring produce the lowest agreement rates and longest negotiations.

Two divergences validate that LLM carriers exercise genuine agency. First, LLM Anchoring and Hardliner carriers agree more often than algorithmic counterparts (63.0\% vs.\ 24.4\% and 52.9\% vs.\ 4.4\%), suggesting the LLM occasionally deviates from strict walk-away thresholds. Second, the broker achieves higher savings against LLM carriers in adversarial personas (Hardliner: 0.658 vs.\ 0.409), indicating slightly larger concessions under equivalent pressure. These differences are expected: behavioral prompts produce stochastic concession curves, providing a more realistic test of the framework's robustness.

\newpage
\section{Offer Curve Comparison}
\label{app:curves}

Figure~\ref{fig:curves} overlays the broker offer curves for all five rule-based strategies on the same load against each of four carrier archetypes (rows), across all three spread regimes (columns). The Cooperative carrier is omitted because all strategies close in round~1, producing no multi-round curves.

Against the Tit-for-Tat carrier (row~1), all strategies converge successfully across all regimes. In the wide regime, Two-Index holds firm alongside Boulware. Against the Hardliner (row~2), agreements are rare; the carrier walks away around round~8 if the broker has not conceded sufficiently. Against the Deadline Exploiter (row~3), the carrier concedes late; agreements cluster in later rounds. Against the Anchoring carrier (row~4), the carrier starts high and concedes slowly; the wide regime again shows Two-Index tracking the Boulware curve.

\begin{figure*}[!htbp]
\centering
\includegraphics[width=\textwidth]{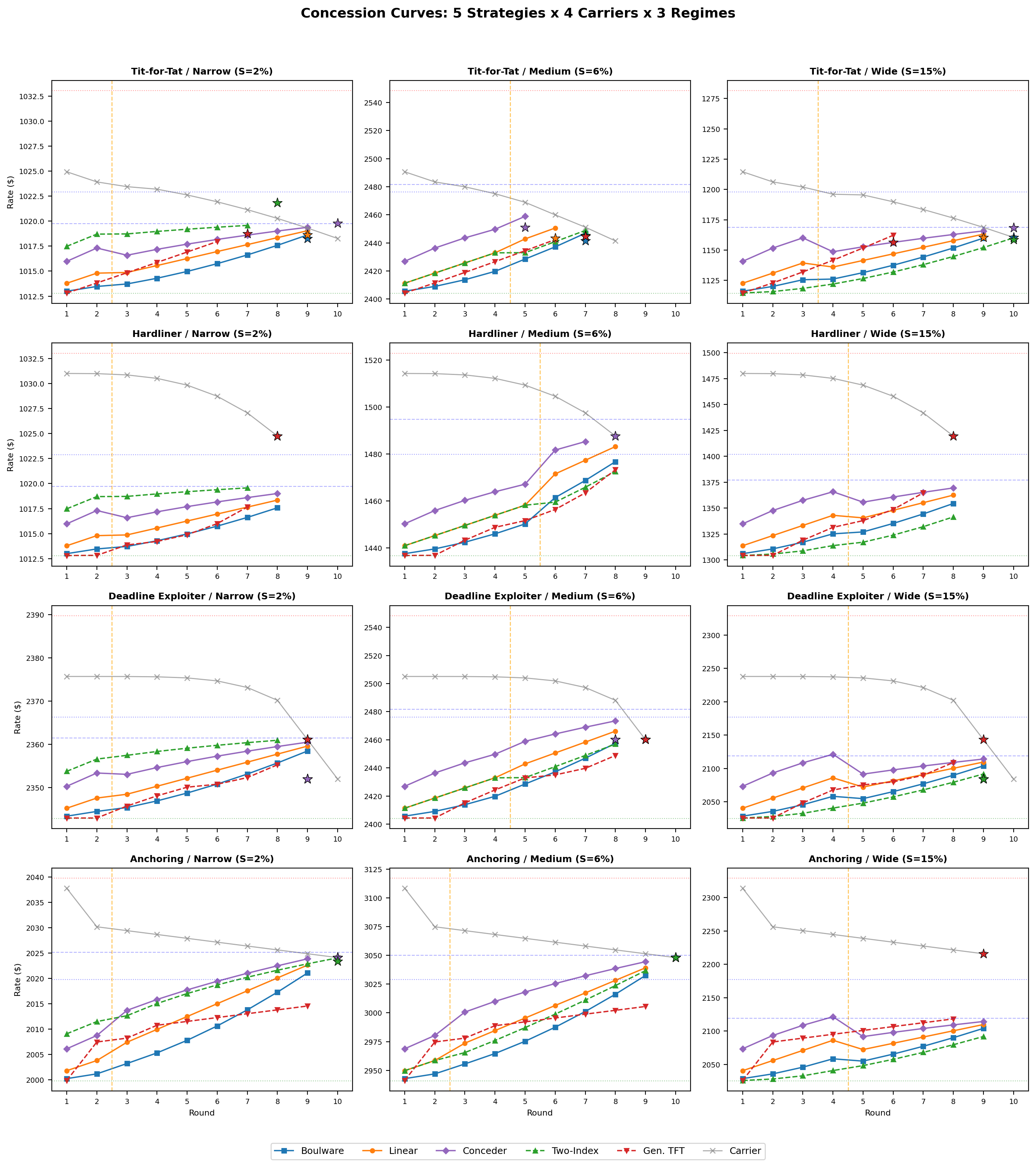}
\caption{Offer curves for all five strategies against each carrier archetype (rows: Tit-for-Tat, Hardliner, Deadline Exploiter, Anchoring) across three spread regimes (columns: narrow $S{=}2\%$, medium $S{=}6\%$, wide $S{=}15\%$). Broker strategies: Boulware $\beta{=}0.6$ (blue squares), Linear $\beta{=}1.0$ (orange circles), Conceder $\beta{=}2.0$ (purple diamonds), Two-Index adaptive $\beta$ (green dashed triangles). Stars mark agreement points. Orange vertical line indicates the dynamic pricing shift.}
\label{fig:curves}
\end{figure*}


\begin{thebibliography}{18}

\bibitem{faratin1998}
Peyman Faratin, Carles Sierra, and Nicholas~R. Jennings.
\newblock Negotiation decision functions for autonomous agents.
\newblock {\em Robotics and Autonomous Systems}, 24(3--4):159--182, 1998.

\bibitem{fatima2014}
Shaheen Fatima, Sarit Kraus, and Michael Wooldridge.
\newblock {\em Principles of Automated Negotiation}.
\newblock Cambridge University Press, 2014.

\bibitem{baarslag2013}
Tim Baarslag, Katsuhide Fujita, Enrico~H. Gerding, Koen Hindriks, Takayuki Ito, Nicholas~R. Jennings, Catholijn Jonker, Sarit Kraus, Raz Lin, Valentin Robu, and Colin~R. Williams.
\newblock Evaluating practical negotiating agents: Results and analysis of the 2011 international competition.
\newblock {\em Artificial Intelligence}, 198:73--103, 2013.

\bibitem{lewis2017}
Mike Lewis, Denis Yarats, Yann Dauphin, Devi Parikh, and Dhruv Batra.
\newblock Deal or no deal? End-to-end learning for negotiation dialogues.
\newblock In {\em Proceedings of EMNLP}, pages 2443--2453, 2017.

\bibitem{he2018}
He~He, Derek Chen, Anusha Balakrishnan, and Percy Liang.
\newblock Decoupling strategy and generation in negotiation dialogues.
\newblock In {\em Proceedings of EMNLP}, pages 2333--2343, 2018.

\bibitem{fu2023}
Yao Fu, Hao Peng, Tushar Khot, and Mirella Lapata.
\newblock Improving language model negotiation with self-play and in-context learning from {AI} feedback.
\newblock {\em arXiv preprint arXiv:2305.10142}, 2023.

\bibitem{bianchi2024}
Federico Bianchi, Patrick~John Chia, Mert Yuksekgonul, Jacopo Tagliabue, Dan Jurafsky, and James Zou.
\newblock How well can {LLMs} negotiate? {NegotiationArena} platform and analysis.
\newblock In {\em Proceedings of ICML}, 2024.

\bibitem{abdelnabi2024}
Sahar Abdelnabi, Amr Gomaa, Sarath Sivaprasad, Lea Schoenherr, and Mario Fritz.
\newblock Cooperation, competition, and maliciousness: {LLM}-stakeholders interactive negotiation.
\newblock In {\em Advances in Neural Information Processing Systems (NeurIPS)}, 2024.

\bibitem{kirshner2026}
Samuel~N. Kirshner, Yiwen Pan, Jason~Xianghua Wu, and Alex Gould.
\newblock Talking terms: Agent information in {LLM} supply chain bargaining.
\newblock {\em Decision Sciences}, 57:9--23, 2026.

\bibitem{lundberg2017}
Scott~M. Lundberg and Su-In Lee.
\newblock A unified approach to interpreting model predictions.
\newblock In {\em Advances in Neural Information Processing Systems (NeurIPS)}, pages 4766--4777, 2017.

\bibitem{cao2015}
Mukun Cao, Xudong Luo, Xin~Robert Luo, and Xiaopei Dai.
\newblock Automated negotiation for e-commerce decision making: A goal deliberated agent architecture for multi-strategy selection.
\newblock {\em Decision Support Systems}, 73:1--14, 2015.

\bibitem{hindriks2008}
Koen Hindriks and Dmytro Tykhonov.
\newblock Opponent modelling in automated multi-issue negotiation using {Bayesian} learning.
\newblock In {\em Proceedings of the 7th International Joint Conference on Autonomous Agents and Multiagent Systems (AAMAS)}, pages 331--338, 2008.

\bibitem{bagga2020}
Pallavi Bagga, Nicola Paoletti, Bedour Alrayes, and Kostas Stathis.
\newblock A deep reinforcement learning approach to concurrent bilateral negotiation.
\newblock In {\em Proceedings of the 29th International Joint Conference on Artificial Intelligence (IJCAI)}, pages 297--303, 2020.

\bibitem{gear2020}
Aditya~Srinivas Gear, Kritika Prakash, Nonidh Singh, and Praveen Paruchuri.
\newblock {PredictRV}: A prediction based strategy for negotiations with dynamically changing reservation value.
\newblock In {\em Group Decision and Negotiation: A Multidisciplinary Perspective}, pages 135--148, Springer, 2020.

\bibitem{nowak1992}
Martin~A. Nowak and Karl Sigmund.
\newblock Tit for tat in heterogeneous populations.
\newblock {\em Nature}, 355(6357):250--253, 1992.

\bibitem{winoto2005}
Pinata Winoto, Gordon McCalla, and Julita Vassileva.
\newblock Non-monotonic-offers bargaining protocol.
\newblock {\em Autonomous Agents and Multi-Agent Systems}, 11(1):45--67, 2005.

\bibitem{endriss2006}
Ulle Endriss.
\newblock Monotonic concession protocols for multilateral negotiation.
\newblock In {\em Proceedings of the 5th International Joint Conference on Autonomous Agents and Multiagent Systems (AAMAS)}, pages 392--399, 2006.

\bibitem{hormuz2026}
Evan Halper and Rachel Siegel.
\newblock Oil prices soar amid worries of sustained war in {I}ran.
\newblock {\em The Washington Post}, March 2, 2026.

\bibitem{mordor2025}
{Mordor Intelligence}.
\newblock {U}nited {S}tates freight brokerage market size \& share analysis: growth trends and forecast (2026 to 2031).
\newblock Mordor Intelligence, 2025.
\newblock \url{https://www.mordorintelligence.com/industry-reports/united-states-freight-brokerage-market}.

\end{thebibliography}
\end{document}